\documentclass[fleqn]{llncs}
\usepackage[T1]{fontenc}
\usepackage{amsfonts,amsmath,amssymb}
\usepackage{graphicx}
\usepackage{multirow}
\usepackage{color}

\newcommand{\red}[1]{{\color{red} #1}}

\def \tuple#1{\langle #1 \rangle}
\def \sem#1{[\![ #1 ]\!]}
\def\defemb#1#2{\expandafter\def\csname #1\endcsname
                              {\relax\ifmmode #2\else\hbox{$#2$}\fi}}
\defemb{cA}{{\mathcal A}}
\defemb{cB}{{\mathcal B}}
\defemb{cC}{{\mathcal C}}
\defemb{cD}{{\mathcal D}}
\defemb{cE}{{\mathcal E}}
\defemb{cF}{{\mathcal F}}
\defemb{cI}{{\mathcal I}}
\defemb{cL}{{\mathcal L}}
\defemb{cO}{{\mathcal O}}
\defemb{cP}{{\mathcal P}}
\defemb{cQ}{{\mathcal Q}}
\defemb{cR}{{\mathcal R}}
\defemb{cS}{{\mathcal S}}
\defemb{cT}{{\mathcal T}}
\defemb{cU}{{\mathcal U}}
\defemb{cV}{{\mathcal V}}
\defemb{cW}{{\mathcal W}}
\newcommand{\ri}{{<\!\!\!<}}                
\newcommand{\var}{{\cV}ar}                
\newcommand{\revto}{{\leftarrow}}         
\newcommand{\sld}{\rightarrow_{AS}}     
\newcommand{\slda}{{\rightarrow_{AS1}}}   
\newcommand{\sldb}{{\rightarrow_{AS2}}}   
\newcommand{\is}{\rightarrow_{IS}}      

\newcommand{\ot}{\leftarrow} 
\newcommand{\adjoint}{\mathbin{\&}}
\newcommand{\comment}[1]{}
\newcommand{\saca}{\textsf{saca}}
\newcommand{\fca}{\textsf{fca}}
\newcommand{\sfca}{\textsf{sfca}}
\newcommand{\sym}{\mathsf{sym}}

\long\def\comment#1{}

\newcommand{\pr}{{\sf Prolog}}
\newcommand{\ma}{{\sf MALP}}
\newcommand{\sma}{{\sf sMALP}}
\newcommand{\fa}{{\sf FASILL}}
\newcommand{\fl}{{\sf FLOPER}}
\newcommand{\bo}{{\sf Bousi$\sim$Prolog}}

\newcommand{\li}{{\sf Likelog}}
\newcommand{\sq}{{\sf SQLP}}

\begin{document}

\title{Tuning Fuzzy Logic Programs\\ with Symbolic Execution%
\thanks{
This work has been partially supported by the EU (FEDER) and the 
Spanish \emph{Ministerio de Econom\'{\i}a y Competitividad} 
under grants TIN2013-45732-C4-2-P, TIN2013-44742-C4-1-R and by the
\emph{Generalitat Valenciana} under grant PROMETEO-II/2015/013
(SmartLogic).
}}

\titlerunning{Tuning Fuzzy Logic Programs While Testing Computed Answers}
\authorrunning{G. Moreno, J. Penabad \& G. Vidal}

\author{Gin\'es Moreno\inst{1}
\and 
Jaime Penabad\inst{2}
\and
Germ\'an Vidal\inst{3}
}
\institute{ Dept. of Computing Systems, UCLM, 
 02071  Albacete (Spain)
 \and
 Dept. of Mathematics, UCLM,  
 02071  Albacete (Spain)
 \and
 MiST, DSIC, Universitat Polit\`ecnica de Val\`encia, Spain  \\
 \email{\{Gines.Moreno,Jaime.Penabad\}@uclm.es}, \email{gvidal@dsic.upv.es}}

\maketitle

\begin{abstract}
Fuzzy logic programming is a growing declarative paradigm aiming
to integrate  fuzzy logic into logic programming. 
One of the most difficult tasks when specifying a fuzzy logic program
is determining the right weights for each rule, as well as the most
appropriate fuzzy connectives and operators. 
In this paper, we introduce a symbolic extension of fuzzy logic
programs in which some of these parameters can be left unknown, so
that the user can easily see the impact of their possible values.
Furthermore, given a number of test cases, the most appropriate values
for these parameters can be automatically computed.
\end{abstract}

\keywordname{ Fuzzy logic programming, symbolic execution, tuning}

\section{Introduction}\label{sec-intro}

\emph{Logic Programming} \cite{Llo87} has been widely used as a formal
method for problem solving and knowledge representation. Nevertheless,
traditional logic programming languages do not incorporate techniques
or constructs to explicitly deal with uncertainty and approximated
reasoning. In order to fill this gap, \emph{fuzzy logic programming}
has emerged as an interesting---and still growing---research area
which aims to consolidate the efforts for introducing fuzzy logic into
logic programming.

During the last decades, several fuzzy logic programming systems have
been developed. Here, essentially, the classical SLD resolution
principle of 
logic programming has been replaced by a fuzzy variant with the aim of
dealing with partial truth and reasoning with uncertainty in a natural
way. Most of these systems implement (extended versions of) the
resolution principle introduced by Lee \cite{Lee72}, such as
\textsf{Elf-Prolog} \cite{IK85}, \textsf{F-Prolog} \cite{LL90}, generalized annotated
logic programming \cite{KS92}, \textsf{Fril} \cite{BMP95}, \ma\ \cite{MOV04},
\textsf{R-fuzzy} \cite{GMV04},
the QLP scheme of \cite{RR08FLOPS} and the many-valued logic
programming language of \cite{Str05,Str08}. There exists also a family
of fuzzy languages based on sophisticated unification methods
\cite{Ses02} which cope with similarity/proximity relations, as occurs
with \li\ \cite{AF99}, \sq~\cite{CRR08}, \bo\
\cite{JRG09ENTCS,RJ14JIFS} and \fa\ \cite{JMPV15EPTCS,JMPV16RULEML}. Some related
approaches based on \emph{probabilistic logic programming} can be
found, e.g., in \cite{NS92,BMG12}. 

In this paper we focus on the so-called \emph{multi-adjoint logic
  programming} approach \ma\ \cite{MOV04,MOV011,MOV012}, a powerful
and promising approach in the area of fuzzy logic programming. Intuitively
speaking, logic programming is extended with a \emph{multi-adjoint
  lattice} $L$ of truth values (typically, a real number between $0$
and $1$), equipped with a collection of \emph{adjoint pairs} $\tuple{
  \&_i, \revto_i}$ and connectives: implications, conjunctions,
disjunctions, and other operators called aggregators, which are
interpreted on this lattice. Consider, for instance, the following
\ma\ rule:
\[
  good(X) ~ \revto_\mathtt{P} ~ @_\mathtt{aver}(nice(X),cheap(X))~with~0.8
\]
where the adjoint pair $\tuple{\&_\mathtt{P},\revto_\mathtt{P}}$ is
defined as 
\[
\begin{array}{llllll}
  \adjoint_{\tt P} (x,y) &\triangleq& x *  y 
  & \quad\quad\quad
    \ot_{\tt P}(x,y) &\triangleq
  & \left\{\begin{array}{ll}
             1  & \text{ if }y\leq x \\[1.5ex]
             x/y & \text{ if } 0<x<y
           \end{array}\right.
\end{array}
\]
and the aggregator $@_\mathtt{aver}$ is typically defined as $ @_{\tt
  aver}(x_1,x_2) \triangleq (x_1 + x_2)/2 $.  Therefore, the rule
specifies that $X$ is good---with a truth degree of 0.8---if $X$ is
nice and cheap. Assuming that $X$ is nice and cheap with, e.g., truth
degrees $n$ and $c$, respectively, then $X$ is good with a truth
degree of $0.8 * ((n+c)/2)$.


When specifying a \ma\ program, sometimes it might be difficult to
assign weights---truth degrees---to program rules, as well as to
determine the right connectives.\!\footnote{For instance, we have
  typically several adjoint pairs: \emph{{{\L{}}}ukasiewicz logic}
  $\tuple{\&_\mathtt{L},\revto_\mathtt{L}}$, \emph{G\"{o}del
    intuitionistic logic} $\tuple{\&_\mathtt{G},\revto_\mathtt{G}}$
  and \emph{product logic} $\tuple{\&_\mathtt{P},\revto_\mathtt{P}}$,
  which might be used for modeling \emph{pessimist}, \emph{optimist}
  and \emph{realistic scenarios}, respectively.} A programmer can
develop a prototype and repeatedly execute it until the set of answers
is the intended one. Unfortunately, this is a tedious and time
consuming operation. Actually, it might be impractical when the
program should correctly model a large number of test cases provided by the
user.

In order to overcome this drawback, in this paper we introduce a
symbolic extension of \ma\ programs called \emph{symbolic
  multi-adjoint logic programming} (\sma). Here, we can write rules
containing \emph{symbolic} truth degrees and \emph{symbolic}
connectives, i.e., connectives which are not defined on its associated
multi-adjoint lattice. In order to evaluate these programs, we
introduce a symbolic operational semantics that delays the evaluation
of symbolic expressions. Therefore, a \emph{symbolic answer} could now
include symbolic (unknown) truth values and connectives. We prove the
correctness of the approach, i.e., the fact that using the symbolic
semantics and then replacing the unknown values and connectives by
concrete ones gives the same result as replacing these values and
connectives in the original \sma\ program and, then, applying the
concrete semantics on the resulting \ma\ program.
Furthermore, we show how \sma\ programs can be used to tune a program
w.r.t.\ a given set of test cases, thus easing what is considered the most
difficult part of the process: the specification of the right weights
and connectives for each rule. We plan to integrate this tuning
process into the \fl\ system ({\em Fuzzy LOgic Programming Environment
  for Research}); see, e.g.,  
\cite{MM08RULEML,MMPV10ICAI,MMPV10RULEML,MV14JSEA}.


The structure of this paper is as follows. After some preliminaries in
Section~\ref{sec-prelims}, we introduce the framework of symbolic
multi-adjoint logic programming in Section~\ref{sec-malp} and prove
its correctness. Then, in Section~\ref{sec-tuning}, we show the
usefulness of symbolic programs for tuning several parameters so that
a concrete program is obtained. Finally, Section~\ref{sec-conc}
concludes and points out some directions for further research.

\section{Preliminaries}  \label{sec-prelims}

We assume the existence of a multi-adjoint lattice $\tuple{L, \preceq,
  \&_1,\revto_1, \ldots, \&_n,\revto_n}$, equipped with a collection
of \emph{adjoint pairs} $\tuple{\&_i,\revto_i}$---where each $\&_i$ is
a conjunctor which is intended to be used for the evaluation of
\emph{modus ponens} \cite{MOV04}---. 
In addition, in the program rules, we can have an adjoint implication 
($\revto_i$), conjunctions (denoted by $\wedge_1,
\wedge_2, \ldots$), adjoint conjuntions ($\&_1,\&_2,\dots$), disjunctions ($|_1, |_2, \ldots$), and other
operators called aggregators (usually denoted by ${@}_1, {@}_2,
\ldots$); see \cite{NW06} for more details.  More exactly,  
a multi-adjoint lattice fulfill the following properties:
\begin{itemize}
\item $\langle L, \preceq \rangle $ is a (bounded) 
complete
  lattice.\!\footnote{A complete lattice is a (partially) ordered set
    $\tuple{L,\preceq}$ such that every subset $S$ of $L$ has infimum
    and supremum elements. It is bounded if it has bottom and top
    elements, denoted by $\bot$ and $\top$, respectively. $L$ is said
    to be the {\it carrier set} of the lattice, and $\preceq$ its
    ordering relation.}
\item For each truth function of $\adjoint_i$, an increase in any of
  the arguments results in an increase of the result (they are
  \emph{increasing}).
\item For each truth function of $\ot_i$, the result increases as the
  first argument increases, but it decreases as the second argument
  increases (they are \emph{increasing} in the consequent and
  \emph{decreasing} in the antecedent).
\item $\tuple{\adjoint_i,\leftarrow_i}$ is an \emph{adjoint pair}
  in $\langle L, \preceq \rangle $, namely, for any $x,y,z \in L$\@, we
  have that: $ x \preceq (y \ot_i z)$ if and only if $(x \adjoint_i z)
  \preceq y$.
\end{itemize}
This last condition, called the \emph{adjoint property}, could be
considered the most important feature of the framework (in contrast
with other approaches) which justifies most of its properties
regarding crucial results for soundness, completeness, applicability,
etc.~\cite{MOV04}.

Aggregation operators are useful to describe or specify user
preferences. An aggregation operator, when interpreted as a truth
function, may be an arithmetic mean, a weighted sum or in general any
monotone function whose arguments are values of a multi-adjoint
lattice $L$.
Although, formally, these connectives are binary operators, we often
use them as $n$-ary functions so that ${@}(x_1, \ldots, {@}(x_{n-1},
x_n),\ldots)$ is denoted by ${@}(x_1, \ldots, x_n)$. By abuse, in
these cases, we consider ${@}$ an $n$-ary operator.
The truth function of an $n$-ary connective $\varsigma$ is denoted by
$\sem{\varsigma}:L^n\mapsto L$ and is required to be monotonic and
fulfill the following conditions:
$\sem{\varsigma}(\top,\dots,\top)=\top$ and
$\sem{\varsigma}(\bot,\dots,\bot)=\bot$.
%

In this work, given a multi-adjoint lattice $L$, we consider a first
order language $\cL_L$ built upon a signature $\Sigma_L$, that
contains the elements of a countably infinite set of variables $\cV$,
function and predicate symbols (denoted by $\cF$ and $\Pi$,
respectively) with an associated arity---usually expressed as pairs
$f/n$ or $p/n$, respectively, where $n$ represents its arity---, and
the truth degree literals $\Sigma_L^T$ and connectives $\Sigma_L^C$
from $L$.
%
%
Therefore, a well-formed formula in $\cL_L$ can be either:
\begin{itemize}
\item A \emph{value} $v\in \Sigma_L^T$, which will be interpreted as
  itself, i.e., as the truth degree $v\in L$.

\item $p(t_1,\ldots,t_n)$, if $t_1,\ldots,t_n$ are terms over
  $\cV\cup\cF$ and $p/n$ is an n-ary predicate. This formula is called
  \emph{atomic} (or just an atom). 

\item $\varsigma(e_1,\ldots,e_n)$, if $e_1,\ldots,e_n$ are well-formed
  formulas and $\varsigma$ is an $n$-ary connective with truth
  function $\sem{\varsigma}:L^n\mapsto L$.
\end{itemize}
%
As usual, a \emph{substitution} $\sigma$ is a mapping
from variables from $\cV$ to terms over $\cV\cup\cF$ such that
$Dom(\sigma) = \{x\in\cV \mid x \neq \sigma(x)\}$ is its
domain. Substitutions are usually denoted by sets of mappings like,
e.g., $\{x_1/ t_1,\ldots,x_n/ t_n\}$. 
Substitutions are extended to morphisms from terms to terms in the
natural way.  The identity substitution is denoted by
$id$. Composition of substitutions is denoted by juxtaposition, i.e.,
$\sigma\theta$ denotes a substitution $\delta$ such that $\delta(X) =
\theta(\sigma(x))$ for all $x\in\cV$.

In the following, an \emph{$L$-expression} is a well-formed formula of
$\cL_L$ which is composed only by values and connectives from $L$,
i.e., expressions over $\Sigma_L^T\cup\Sigma_L^C$.

In what follows, we assume that the truth function of any connective
$\varsigma$ in $L$ is given by a corresponding definition of the
form ${\sem{\varsigma}}(x_1, \ldots, x_n)\triangleq E$.\footnote{For 
convenience, in the following sections, we not distinguish between the 
connective $\varsigma$ and its truth function $\sem{ \varsigma}$.} 
For instance, in this work, we will be mainly concerned with the 
classical
set of adjoint pairs (conjunctors and implications) over
$\tuple{[0,1],\leq}$ shown in Figure \ref{fig-adj}, 
where labels $\tt L$, $\tt G$ and $\tt P$ mean
respectively \emph{{{\L{}}}ukasiewicz logic}, \emph{G\"{o}del
  intuitionistic logic} and \emph{product logic} (which might be used
for modeling \emph{pessimist}, \emph{optimist} and \emph{realistic
  scenarios}, respectively).

\begin{figure}[t]
\label{fig-adj}
$\begin{array}{lllllllllllllrl}
  \adjoint_{\tt P} (x,y) &\triangleq& x *  y & \quad
  \ot_{\tt P}(x,y) &\triangleq& \left\{\begin{array}{cl}
    1  & \text{ if }y\leq x \\[1.5ex]
    x/y & \text{ if } 0<x<y
  \end{array}\right. & ~~~   
  \hbox{\emph{Product logic}}
  \\[1.2ex]
  \adjoint_{\tt G} (x,y) & \triangleq & \min(x,y) &  \quad
  \ot_{\tt G}(x,y) &\triangleq &  \left\{\begin{array}{cl}
    1  & \text{ if }y\le x \\
    x & \text{otherwise}
  \end{array}\right. & ~~~ \hbox{\emph{G\"odel logic}}
\\[1.2ex]
\adjoint_{\tt L} (x,y) & \triangleq &\max(0,x+y-1)  &  \quad
\ot_{\tt L}(x,y) & \triangleq & \min (x-y+1,1) & ~~~ 
\hbox{\emph{\L ukasiewicz logic}}\\
\end{array}\\$
 \caption{Adjoint pairs of three different fuzzy logics over
$\tuple{[0,1],\leq}$.}
\end{figure}
\comment{
\[
\begin{array}{lllllllllllllrl}
  \adjoint_{\tt P} (x,y) &\triangleq& x *  y & \quad
  \ot_{\tt P}(x,y) &\triangleq& \min(1,x/y) & ~~~
  \hbox{\emph{Product}}
  \\
  \adjoint_{\tt G} (x,y) & \triangleq & \min(x,y) &  \quad
  \ot_{\tt G}(x,y) &\triangleq &  \begin{cases}
    1  & \text{ if }y\le x \\
    x & \text{otherwise}
  \end{cases} & ~~~ \hbox{\emph{G\"odel }}
\\
\adjoint_{\tt L} (x,y) & \triangleq &\max(0,x+y-1)  &  \quad
\ot_{\tt L}(x,y) & \triangleq & \min \{x-y+1,1\} & ~~~ 
\hbox{\emph{\L ukasiewicz}}\\
\end{array}
\]
}
A \ma\ \emph{rule} over a multi-adjoint lattice $L$ is a formula
$H~\revto_{i}~\cB$, where $H$ is an \emph{atomic formula} (usually
called the \emph{head} of the rule), $\revto_{i}$ is an implication
symbol belonging to some adjoint pair of $L$, and $\cB$ (which is
called the \emph{body} of the rule) is a well-formed formula over $L$
without implications.
A \emph{goal} is a body submitted as a query to the
system. 
A \ma\ program is a set of expressions 
$R ~with~ v$,
where $R$ is a rule and $v$ is a \emph{truth degree} (a value of $L$) expressing the
confidence of a programmer in the truth of rule $R$. By abuse of the
language, we often refer to $R ~with~ v$ as a rule.

See, e.g., \cite{MOV04} for a complete formulation of the \ma\
framework.

\section{Symbolic Multi-adjoint Logic Programming}\label{sec-malp}

In this section, we introduce a \emph{symbolic} extension of
multi-adjoint logic programming.  Essentially, we will allow some
undefined values (truth degrees) and connectives in the program
rules, so that these elements can be systematically computed
afterwards.
In the following, we will use the abbreviation \sma\ to refer to
programs belonging to this setting.

Here, given a multi-adjoint lattice $L$, we consider an augmented
language $\cL_L^s\supseteq\cL_L$ which may also include a number of
symbolic values, symbolic adjoint pairs and symbolic connectives which
do not belong to $L$. Symbolic objects are usually denoted as $o^s$
with a superscript $s$.

\begin{definition}[\sma\ program]
  Let $L$ be a multi-adjoint lattice.
  An \sma\ program over $L$ is a set of {symbolic rules}, where
  each symbolic rule is a formula $(H~\revto_i~\cB~ with~v)$, where
  the following conditions hold:
  \begin{itemize}
  \item $H$ is an atomic formula of $\cL_L$ (the head of the rule);
  \item $\revto_i$ is a (possibly symbolic) implication from either a
    symbolic adjoint pair $\tuple{\&^s,\revto^s}$ or from an adjoint
    pair of $L$;
  \item $\cB$ (the body of the rule) is a symbolic goal, i.e., a
    well-formed formula of $\cL_L^s$;
  \item $v$ is either a truth degree (a value of $L$) or a symbolic
    value.
  \end{itemize}
\end{definition}
%

\begin{example} \label{ex1}
  We consider the multi-adjoint lattice $\tuple{[0,1],\leq,\&_\mathtt{P},
    \revto_\mathtt{P},\&_\mathsf{G},\revto_\mathsf{G},\&_\mathtt{L},$ $\revto_\mathtt{L}}$,
  where the adjoint pairs are defined in Section~\ref{sec-prelims},
  also including 
	$@_{\tt aver}$ which is defined as follows: ${@_{\tt aver}}(x_1,x_2) \triangleq (x_1 +
  x_2)/2$. Then, the following is an \sma\ program $\cP$:
  \[
  \begin{array}{l@{~~}l@{~~}l@{~~}ll@{~~}l}
    & p(X) & \revto^{s_1} & \&^{s_2}(q(X),@_{\tt aver}(r(X),s(X)))
    & \mathit{with} & 0.9 \\
    & q(a) & & &   \mathit{with} &   v^s\\
    & r(X) & &  & \mathit{with}  & 0.7\\
    & s(X) & &  & \mathit{with} &  0.5
  \end{array}
  \]
  where $\tuple{\&^{s_1},\revto^{s_1}}$ is a symbolic adjoint pair
  (i.e., a pair not defined in $L$), $\&^{s_2}$ is a symbolic
  conjunction, and $v^s$ is a symbolic value.
\end{example}
The procedural semantics of 
\sma\ is defined in a stepwise manner as follows. First, an
\emph{operational} stage is introduced which proceeds similarly to SLD
resolution in pure logic programming. In contrast to standard logic
programming, though, our operational stage returns an expression still
containing a number of (possibly symbolic) values and
connectives. Then, an \emph{interpretive} stage evaluates these
connectives and produces a final answer (possibly containing symbolic
values and connectives).

In the following, $\cC[A]$ denotes a formula where $A$ is a
sub-expression which occurs in the---possibly empty---context
$\cC[]$. Moreover, $\cC[A/A']$ means the replacement of $A$ by $A'$ in
context $\cC[]$, whereas $\var(s)$ refers to the set of distinct
variables occurring in the syntactic object $s$, and $\theta
[\var(s)]$ denotes the substitution obtained from $\theta$ by
restricting its domain to $\var(s)$. An \sma\ \emph{state} has the
form $\tuple{\cQ;\sigma}$ where $\cQ$ is a symbolic goal and $\sigma$
is a substitution. We let $\cE^s$ denote the set of all possible \sma\
states.

\begin{definition}[admissible step]\label{as}
  Let $L$ be a multi-adjoint lattice and $\cP$ an \sma\ program over
  $L$. An \emph{admissible step} is formalized as a state transition
  system, whose transition relation $\sld{} \subseteq (\cE^s \times
  \cE^s)$ is the smallest relation satisfying the following transition
  rules:\footnote{Here, we assume that $A$ in $\cQ[A]$ is the selected
    atom. Furthermore, as 
		common practice, $mgu(E)$ denotes the
    \emph{most general unifier} of the set of equations $E$
    \cite{LMM88}.}
  \begin{enumerate}
  \item $\tuple{\cQ[A];\sigma} ~\sld~
    \tuple{(\cQ[A/v\&_i\cB])\theta;\sigma\theta}$,\\
    if $\theta = mgu(\{H= A\})\neq fail$, $(H~\revto_i~ \cB ~ with ~
    v) \ri \cP$ and $\cB$ is not empty.\footnote{For simplicity, we
      consider that facts $(H ~ with ~v)$ are seen as rules of the
      form $(H\revto_i \top ~with ~v)$ for some implication
      $\revto_i$. Furthermore, in this case, we directly derive
      the state $\tuple{(\cQ[A/v])\theta;\sigma\theta}$ since
      $v\:\&_i\top = v$ for all $\&_i$.  }
  \item $\tuple{\cQ[A];\sigma} ~ \sld~
    \tuple{(\cQ[A/\bot]);\sigma}$, \\
    if there is no rule $(H~\revto_i~ \cB ~ with ~ v) \ri \cP$ such
    that $mgu(\{H=A\})\neq fail$.
  \end{enumerate}
  Here, $(H~\revto_i~ \cB ~ with ~ v) \ri \cP$ denotes that
  $(H~\revto_i~ \cB ~ with ~ v)$ is a renamed apart variant of a rule
  in $\cP$ (i.e., all its variables are fresh). Note that symbolic
  values and connectives are not renamed.
\end{definition}
Observe that the second rule is needed to cope with expressions like \linebreak 
$@_{aver}(p(a),0.8)$, which can be evaluated successfully even when there
is no rule matching $p(a)$ since $@_{aver}(0,0.8)=0.4$.


In the following, given a relation $\to$, we let $\to^\ast$ denote its
reflexive and transitive closure. Also, an \emph{$L^s$-expression} is
now a well-formed formula of $\cL_L^s$ which is composed by values and
connectives from $L$ as well as by symbolic values and connectives.

\begin{definition}[admissible derivation] \label{aca} Let $L$ be a
  multi-adjoint lattice and $\cP$ be an \sma\ program over $L$. Given
  a goal $\cQ$, an \emph{admissible derivation} is a sequence
  $\tuple{\cQ; id}\sld^\ast \tuple{\cQ';\theta}$. When $\cQ'$ is an
  $L^s$-expression, the derivation is called \emph{final} and the pair
  $\tuple{\cQ';\sigma}$, where $\sigma=\theta[\var(\cQ)]$, is called a
  \emph{symbolic admissible computed answer} (\saca, for short) for
  goal $\cQ$ in $\cP$.
\end{definition}


\begin{example} \label{exa-lfp} Consider again the multi-adjoint
  lattice $L$ and the \sma\ program $\cP$ of Example~\ref{ex1}.  Here,
  we have the following final admissible derivation for $p(X)$ in
  $\cP$ (the selected atom is underlined):
  \[
  \begin{array}{lll}
    \tuple{\underline{p(X)}; ~id}
    & \sld &  \tuple{\&^{s_1}(0.9, \&^{s_2}
      (\underline{q(X_1)},@_{\tt aver}(r(X1),s(X1)))); \{X/X_1\} } \\
    & \sld &
    \tuple{\&^{s_1}(0.9, \&^{s_2}
      (v^s,@_{\tt aver}(\underline{r(a)},s(a)))); \{X/a,X_1/a\} } \\
    & \sld &
    \tuple{\&^{s_1}(0.9, \&^{s_2}
      (v^s,@_{\tt aver}(0.7,\underline{s(a)}))); \{X/a,X_1/a,X_2/a\} } \\
    & \sld &
    \tuple{\&^{s_1}(0.9, \&^{s_2}
      (v^s,@_{\tt aver}(0.7,0.5))); \{X/a,X_1/a,X_2/a,X_3/a\} }
  \end{array}
  \]
  Therefore, the associated \saca\ is $\tuple{\&^{s_1}(0.9, \&^{s_2}
    (v^s,@_{\tt aver}(0.7,0.5))); \{X/a\} }$.
\end{example}
Given a goal $\cQ$ and a final admissible derivation $\tuple{\cQ;id}
\sld^\ast \tuple{\cQ';\sigma}$, we have that $\cQ'$ does not contain
atomic formulas. Now, $Q'$ can be \emph{solved} by using the following
interpretive stage:
%

\begin{definition}[interpretive step]\label{is}
  Let $L$ be a multi-adjoint lattice and $\cP$ be an \sma\ program
  over $L$. Given a \saca\ $\tuple{Q;\sigma}$, the \emph{interpretive}
  stage is formalized by means of the following transition relation
  $\is \subseteq (\cE^s \times \cE^s)$, which is defined as the least
  transition relation satisfying:
  \[
  \tuple{\cQ[\varsigma(r_1,\ldots,r_n)];\sigma} ~\is~
  \tuple{\cQ[\varsigma(r_1,\dots,r_n)/r_{n+1}]; \!\sigma} 
  \]
  where $\varsigma$ denotes a connective defined on $L$ and
  $\sem{\varsigma}\!(r_1,\ldots,r_n)=r_{n+1}$.
  
  An interpretive derivation of the form
  $\tuple{\cQ;\sigma}\is^\ast\tuple{\cQ';\theta}$ such that
  $\tuple{\cQ';\theta}$ cannot be further reduced, is called a
  \emph{final} interpretive derivation. In this case,
  $\tuple{\cQ';\theta}$ is called a \emph{symbolic fuzzy computed
    answer} (\sfca, for short).
  Also, if $\cQ'$ is a value of $L$, we say that $\tuple{\cQ';\theta}$
  is a fuzzy computed answer (\fca, for short).
\end{definition}

\begin{example} \label{exa-ica} Given the \saca\ of Ex.~\ref{exa-lfp}:
  $\tuple{\&^{s_1}(0.9, \&^{s_2} (v^s,@_{\tt aver}(0.7,0.5))); \{X/a\}
  }$, we have the following final interpretive derivation (the
  connective reduced is underlined):
  \[
  \tuple{\&^{s_1}(0.9, \&^{s_2} (v^s,\underline{@_{\tt
        aver}(0.7,0.5)})); \{X/a\} } \is \tuple{\&^{s_1}(0.9, \&^{s_2}
    (v^s,0.6)); \{X/a\} }
  \]
  with $\sem{@_\mathtt{aver}}(0.7,0.5)=0.6$.
  Therefore, $\tuple{\&^{s_1}(0.9, \&^{s_2} (v^s,0.6)); \{X/a\} }$ is
  a \sfca\ of $p(X)$ in $\cP$.
\end{example}
Given a multi-adjoint lattice $L$ and a symbolic language $\cL_L^s$,
in the following we consider \emph{symbolic substitutions} that are
mappings from symbolic values and connectives to expressions over
$\Sigma_L^T\cup\Sigma_L^C$. Symbolic substitutions are denoted by
$\Theta,\Gamma,\ldots$ Furthermore, for all symbolic substitution
$\Theta$, we require the following condition: $\revto^s/\revto_i
\in\Theta$ iff $\&^s/\&_i\in\Theta$, where $\tuple{\&^s,\revto^s}$ is
a symbolic adjoint pair and $\tuple{\&_i,\revto_i}$ is an adjoint pair
in $L$. Intuitively, this is required for the substitution to have the
same effect both on the program and on an $L^s$-expression.

Given an \sma\ program $\cP$ over $L$, we let $\sym(\cP)$ denote the
symbolic values and connectives in $\cP$. Given a symbolic
substitution $\Theta$ for $\sym(\cP)$, we denote by $\cP\Theta$ the
program that results from $\cP$ by replacing every symbolic symbol
$e^s$ by $e^s\Theta$. Trivially, $\cP\Theta$ is now a \ma\ program.

The following theorem is our key result in order to use \sma\ programs
for tuning the components of a \ma\ program:

\begin{theorem} \label{th1}
  Let $L$ be a multi-adjoint lattice and $\cP$ be an \sma\ program
  over $L$. Let $\cQ$ be a goal. Then, for any symbolic substitution
  $\Theta$ for $\sym(\cP)$, we have that $\tuple{v;\theta}$ is a \fca\
  for $Q$ in $\cP\Theta$ iff there exists a \sfca\ $\tuple{Q';\theta'}$
  for $\cQ$ in $\cP$ and $\tuple{\cQ'\Theta;\theta'} \is^\ast
  \tuple{v;\theta'}$, where $\theta'$ is a renaming of $\theta$.
\end{theorem}

\begin{proof} (Sketch)
  For simplicity, we consider that the same fresh variables are used
  for renamed apart rules in both derivations.

  Consider the following derivations for goal $\cQ$ w.r.t. programs
  $\cP$ and $\cP\Theta$, respectively:
  \[
\begin{array}{l}
  \cD_{\cP~~}: \tuple{\cQ;id} \sld^\ast \tuple{\cQ'';\theta}
  ~~\is^\ast \tuple{\cQ';\theta}
  \\
  \cD_{\cP\Theta}: \tuple{\cQ;id} \sld^\ast
  \tuple{\cQ''\Theta;\theta} \is^\ast \tuple{\cQ'\Theta;\theta}
  \end{array}
  \] 
  Our proof proceeds now in three stages:
  \begin{enumerate}
  \item Firstly, observe that the sequences of symbolic admissible
    steps in $\cD_{\cP}$ and $\cD_{\cP\Theta}$ exploit the whole set
    of atoms in both cases, such that a program rule $R$ is used in
    $\cD_{\cP}$ iff the corresponding rule $R\Theta$ is applied in
    $\cD_{\cP\Theta}$ and hence, the \saca's of the derivations are
    $\tuple{\cQ'';\theta}$ and $\tuple{\cQ''\Theta;\theta}$,
    respectively.
  \item Next, we proceed by applying interpretive steps till reaching
    the \sfca\ $\tuple{\cQ';\theta}$ in the first derivation
    $\cD_{\cP}$ and it is easy to see that the same sequence of
    interpretive steps are applied in $\cD_{\cP\Theta}$ thus leading
    to state $\tuple{\cQ'\Theta;\theta}$, even when in this last case
    we do not necessarily obtain a \sfca.
  \item Finally, it suffices by instantiating the \sfca\
    $\tuple{\cQ';\theta}$ in the first derivation $\cD_{\cP}$ with the
    symbolic substitution $\Theta$, for completing both derivations
    with the same sequence of interpretive steps till reaching the
    desired \fca\ $\tuple{v;\theta}$.
  \end{enumerate}
  \qed
\end{proof}

\begin{example}
  Consider again the multi-adjoint lattice $L$ and the \sma\ program
  $\cP$ of Example~\ref{ex1}. Let $\Theta =
  \{\revto^{s_1}/\revto_\mathtt{P},\&^{s_1}/\&_\mathtt{P},
  \&^{s_2}/\&_\mathtt{G}, v^s/0.8\}$ be a symbolic substitution.
  Given the \sfca\ from Example~\ref{exa-ica}, we have
  \[
  \tuple{\&^{s_1}(0.9, \&^{s_2} (v^s,0.6))\Theta; \{X/a\} } =
  \tuple{\&_\mathtt{P}(0.9, {\&_\mathtt{G} (0.8,0.6)});
    \{X/a\} }
  \]
  Therefore, we have the following interpretive final derivation for
  the the instantiated \sfca:
  \[
  \begin{array}{lll}
    \tuple{\&_\mathtt{P}(0.9, \underline{\&_\mathtt{G}
        (0.8,0.6)}); \{X/a\} } \is \tuple{\underline{\&_\mathtt{P}(0.9,0.6)}; \{X/a\} }
    \is \tuple{0.54; \{X/a\} }
    \end{array}
  \]
  By Theorem~\ref{th1}, we have that $\tuple{0.54; \{X/a\} }$ is also
  a \fca\ for $p(X)$ in $\cP\Theta$.
\end{example}

\section{Tuning Multi-adjoint Logic Programs}\label{sec-tuning}

In this section, we introduce an automated technique for tuning
multi-adjoint logic programs using \sma\ programs. We illustrate the
technique with an example.

Consider a typical \pr\ clause ``$H :- B_1,\ldots,B_n$''. It can be
fuzzified in order to become a \ma\ rule ``$H ~\revto_{label}~
\cB \mbox{ with } v$'' by performing the following actions:
\begin{enumerate}
\item weighting it with a truth degree $v$,
\item connecting its head and body with a fuzzy implication symbol
  $\revto_{label}$ (belonging to a concrete adjoint pair
  $\tuple{\revto_{label}, \&_{label}}$) and,
\item linking the set of atoms $B_1,\ldots,B_n$ on its body $\cB$ by
  means of a set of fuzzy connectives (i.e., conjunctions $\&_i$,
  disjunctions $|_j$ or aggregators $@_k$).
\end{enumerate}
Introducing changes on each one of the three fuzzy components just
described above may affect---sometimes in an unexpected way---the set
of fuzzy computed answers for a given goal.  

Typically, a programmer has a model in mind where some parameters have
a clear value. For instance, the truth value of a rule might be
statistically determined and, thus, its value is easy to obtain. In
other cases, though, the most appropriate values and/or connectives 
depend on some subjective notions and, thus, programmers do not 
know how to obtain these values. In a typical scenario, we have
an extensive set of \emph{expected} computed answers (i.e., \emph{test cases}), so the
programmer can follow a ``try and test'' strategy. Unfortunately, this
is a tedious and time consuming operation. Actually, it might even be
impractical when the program should correctly model a large number of
test cases provided by the user.

Therefore, we propose an automated technique that proceeds as
follows. In the following, for simplicity, we only consider the first
answer to a goal. Note that this is not a significant restriction
since one can, e.g., encode multiple solutions in a list so that the
main goal is always deterministic and all non-deterministic calls are
hidden in the computation. Extending the following algorithm for
multiple solutions is not difficult, but makes the formalization more
cumbersome. Hence, we say that a \emph{test case} is a pair $(Q,f)$
where $Q$ is a goal and $f$ is a \fca.

\begin{definition}[algorithm for symbolic tuning of \ma\ programs] \label{alg}
\begin{description}
\item[\textbf{Input:}] an \sma\ program $\cP^s$ and a number of
  (expected) test cases $(Q_i,\tuple{v_i;\theta_i})$, where $Q_i$ is a
  goal and $\tuple{v_i;\theta_i}$ is its expected \fca\  for
  $i=1,\ldots,k$.
\item[\textbf{Output:}] a symbolic substitution $\Theta$.
\end{description}

\begin{enumerate}
\item For each test case $(Q_i,\tuple{v_i;\theta_i})$, compute the
  \sfca\ $\tuple{Q'_i,\theta_i}$ of $\tuple{Q_i,id}$ in $\cP^s$.
\item Then, consider a finite number of possible symbolic
  substitutions for $\sym(\cP^s)$, say $\Theta_1,\ldots,\Theta_n$,
  $n>0$.
\item For each $j\in\{1,\ldots,n\}$, compute
  $\tuple{Q'_i\Theta_j,\theta_i} \is^\ast \tuple{v_{i,j};\theta_i}$,
  for $i=1,\ldots,k$. Let $d_{i,j} = |v_{i,j} - v_i|$, where $|\_|$
  denotes the absolute value.
\item Finally, return the symbolic substitution $\Theta_j$ that
  minimizes $\sum_{i=1}^k d_{i,j}$.
\end{enumerate}
\end{definition}
Observe that the precision of the algorithm can be parameterized
depending on the set of symbolic substitutions considered in step
(2). For instance, one can consider only truth values
$\{0.3,0.5,0.8\}$ or a larger set
$\{0.1,0.2,\ldots,1.0\}$; one can consider only
three possible connectives, or a set including ten of them. Obviously, the
larger the domain of values and connectives is, the more precise the
results are (but the algorithm is more expensive, of course).

%
This algorithm represents a much more efficient method for tuning the
fuzzy parameters of a \ma\ program than repeatedly executing the
program from scratch.


Let us explain the technique by means of a small, but realistic
example.

Here, we assume that a travel agency offers booking services on a
large number of hotels. The travel agency has a web site where the
user can rate every hotel with a value between 1\% and 100\%. The
purpose in this case is to specify a fuzzy model that correctly
represents the rating of each hotel.

In order to simplify the presentation, we consider that there are only
three hotels, named {\it sun}, {\it sweet} and {\it lux}. In the web
site, these hotels have been rated $0.60$, $0.77$ and $0.85$
(expressed as real numbers between $0$ and $1$), respectively.  Our
simple model just depends on three factors: the hotel facilities, the
convenience of its location, and the rates, denoted by predicates
$\mathit{facilities}$, $\mathit{location}$ and $\mathit{rates}$,
respectively. An \sma\ program modelling this scenario is the
following:

%
\[  \begin{array}{l@{~~}l@{~~}l@{~~}ll@{~~}l}
     \mathit{popularity}(X) & \revto^s & |^s(\mathit{facilities}(X),@_{\tt aver}(\mathit{location}(X),\mathit{rates}(X)))
    & \mathit{with} & 0.9 \\
		\\
     \mathit{facilities}(sun) & & &	\mathit{with} &   v^s\\
     \mathit{location}(sun) &  & &		\mathit{with}  & 0.4\\
     \mathit{rates}(sun) &  &  &		\mathit{with} &  0.7\\
		\\
     \mathit{facilities}(sweet) & & &	\mathit{with} &   0.5\\
     \mathit{location}(sweet) &  & &		\mathit{with}  & 0.3\\
     \mathit{rates}(sweet) &  &  &		\mathit{with} &  0.1\\
		\\
     \mathit{facilities}(lux) & & &	\mathit{with} &   0.9\\
     \mathit{location}(lux) &  & &		\mathit{with}  & 0.8\\
     \mathit{rates}(lux) &  &  &		\mathit{with} &  0.2\\
  \end{array}
\]
Here, we assume that all weights can be easily obtained except for the
weight of the fact $\mathit{facilities}(sun)$, which is unknown, so we
introduce a symbolic weight $v^s$. Also, the programmer has some
doubts on the connectives used in the first rule, so she introduced a
number of symbolic connectives: the implication and disjunction
symbols, i.e. $\revto^s$ and $|^s$.

\begin{figure}[t]
$\begin{array}{lllllllllllllrl}
 ~~~~&  \adjoint_{\tt P} (x,y) &=& x *  y & ~~~~~~  
  |_{\tt P}(x,y) &=& x+y-x*y & ~~~~~
    \emph{Product logic}
\\ 
&  \adjoint_{\tt G} (x,y) &=& \min(x,y) & ~~~~~~  
  |_{\tt G}(x,y) &=& \max(x,y) & ~~~~~
      \emph{G\"odel logic}
\\ 
 & \adjoint_{\tt L} (x,y) &=&\max(x+y-1,0) & ~~~~~~  
  |_{\tt L}(x,y) &=& \min(x+y,1) & ~~~~~
    \emph{{{\L{}}ukasiewicz logic} }\\ 
  \end{array}\\$	
 \caption{Conjunctions and disjunctions of three different fuzzy logics over
$\tuple{[0,1],\leq}$.} \label{fig-conn}

\end{figure}

We consider, for each symbolic connective, the three
possibilities shown in Figure~\ref{fig-conn} over the lattice
$\tuple{[0,1],\leq}$, which are based in the so-called {\it Product},
{\em G\"{o}del} and {\em {{\L{}}}ukasiewicz} fuzzy logics.  Adjectives
like {\em pessimist}, {\em realistic} and {\em optimist} are sometimes
applied to the {\em {{\L{}}}ukasiewicz}, {\em Product} and {\em
  G\"{o}del} fuzzy logics, respectively, since conjunctive operators
satisfy that, for any pair of real numbers $x$ and $y$ in $[0,1]$, we
have:

$$
0
 \leq
  \&_{\tt L}(x,y)
 \leq
  \&_{\tt P} (x,y)
  \leq
  \&_{\tt G} (x,y)
  \leq
 1
$$

\noindent In contrast, the contrary holds for the disjunction operations, that is:

$$
0
 \leq
  |_{\tt G}(x,y)
 \leq
  |_{\tt P} (x,y)
  \leq
  |_{\tt L} (x,y)
  \leq
 1
$$

\noindent Note that it is more difficult to satisfy a condition based
on a pessimist conjunction/disjunction (i.e, inspired by the {\em
  {{\L{}}}ukasiewicz} and {\em G\"{o}del} fuzzy logics, respectively)
than with {\em Product} logic based operators, while the optimistic
versions of such connectives
are less restrictive, obtaining greater truth degrees on \fca's. This is
a consequence of the following chain of inequalities:
$$ 0
 \leq
  \&_{\tt L}(x,y)
 \leq
  \&_{\tt P} (x,y)
  \leq
  \&_{\tt G} (x,y)
  \leq
  |_{\tt G}(x,y)
 \leq
  |_{\tt P}(x,y)
  \leq
  |_{\tt L}(x,y)
  \leq
 1
$$

\noindent Therefore, it is desirable to tune the symbolic constants $\revto^s$ and $|^s$ in the first rule of our symbolic \sma\ program  by selecting operators in the previous sequence 
until finding solutions satisfying in a stronger (or weaker) way the user's requirements. 

Focusing on our particular \sma\ program, we consider the following 
three test cases:   
\[
\begin{array}{l}
  (popularity(sun),\tuple{ 0.60; id }),\\
  (popularity(sweet), \tuple{0.77; id} ),\\
  (popularity(lux),\tuple{0.85;id})
\end{array}
\]
for which the respective three \sfca's achieved after applying the first step of our tuning algorithm are: 
\[\begin{array}{lll}
\tuple{ \&^s(0.9,|^s(v^s,0.55)); id 
}\\
\tuple{ \&^s(0.9,|^s(0.5,0.65)); id 
}\\
\tuple{ \&^s(0.9,|^s(0.9,0.5)); id 
}\\
\end{array}
\]
\comment{
R4 < &luka(0.9,|prod(0.3,@aver(0.4,0.7))), {X/sun,X1/sun} >

R7 < &luka(0.9,|prod(0.5,@aver(0.3,1))), {X/sweet,X1/sweet} >

R10 < &luka(0.9,|prod(0.9,@aver(0.8,0.2))), {X/lux,X1/lux} >
}
In the second step of the algorithm, 
we must provide symbolic substitutions for being applied to this set
of \sfca's in order to transform them into \fca's which are as close as possible to those in the test cases.  
Table~\ref{tabla}  
shows the results of the tuning process, where each column has the following meaning:
\begin{itemize}
\item The first pair of columns serve for choosing the implication\footnote{It is important to note that, at execution time, each implication symbol belonging to a concrete adjoint pair is replaced by its adjoint conjunction (see again our repertoire of adjoint pairs in Figure \ref{fig-adj} in the preliminaries section).}  and disjunction connectives of the first program rule (i.e., $\revto^s$ and $|^s$) from each one of the three fuzzy logics considered so far.
\item In the third column, we consider three possible truth degrees
  ($0.3$, $0.5$ and $0.7$) as the potential assignment to the symbolic
  weight $v^s$. In this example, this set suffices to obtain an
  accurate solution.  
\item Each row represents a different symbolic substitution, which are
  shown in column four. 
\item Next, headed by the name of each hotel in the test cases, we have pairs of columns which represent, respectively, the potential truth degree associated to the \fca\ obtained with the corresponding symbolic substitution, 
and  the deviation of such  value w.r.t. the expected truth degree, thus summarizing the computations performed on the third step of our algorithm.     
\item The sum of the three deviations is expressed in the last column of the table, which conforms the value to be minimized as indicated in the final, fourth step of the algorithm.
\end{itemize}
According to these criteria, we observe that the cell with lower value (in particular, $0.05$) in the last column of Table $1$  
refers to the symbolic substitution 
$\Theta_{4} = \{\revto^{s}/\revto_\mathtt{L},|^{s}/|_\mathtt{P}, v^s/0.3\}$\footnote{Anyway, 
another interesting alternative, with a slightly greater deviation---0.6---, could be 
$\Theta_{13}=\{\revto^{s}/\revto_\mathtt{P},|^{s}/|_\mathtt{P},v^s/0.3\}$.}, 
which solve our tuning problem by suggesting that the first pair of
rules in our final, tuned \ma\ program should be the following ones:
\[  \begin{array}{l@{~~}l@{~~}l@{~~}ll@{~~}l}
     rating(X) & \revto_\mathtt{L} & |_\mathtt{P}(services(X),@_{\tt aver}(location(X),budget(X)))
    & \mathit{with} & 0.9 \\
     services(sun) & & &	\mathit{with} &   0.3\\
		\end{array}
\]
%


\begin{table}[!t]\label{tab}
\caption{Table summarizing the results achieved when tuning
  connectives and weights.} \label{tabla}
\centering\scriptsize
~~\\
\begin{tabular}{|c||c|c|c|c|c|c|c|c|c|c|}
\hline
$\revto^s$ & $
|^s$&~$v^s~$ &$\Theta$&\multicolumn{2}{c|}{sun~}&\multicolumn{2}{c|}{sweet~}&\multicolumn{2}{c|}{lux~}&z\\[1.2ex]
\hline
\hline
\multirow{9}{0.6in}{$\revto_{L}$}&\multirow{3}{0.5in}{$|_{G}$} &~0.3~&~$\Theta_1$~&~0.45~&0.15&~0.55&0.22&~0.80&0.05&~0.42~ \\[1ex]\cline{3-11}
                                 &                             &~0.5~&~$\Theta_2$~&~0.45~&0.15&~0.55&0.22&~0.80&0.05&~0.42~ \\[1ex]\cline{3-11}
																 &                             &~0.7~&~$\Theta_3$~&~0.60~&0.00&~0.55&0.22&~0.80&0.05&~0.27~ \\[1ex]\cline{2-11}
																 &\multirow{3}{0.5in}{$|_{P}$} &~\textbf{0.3}~&~$\mathbf{\Theta_4}$~&~\textbf{0.59}~&\textbf{0.01}&~\textbf{0.73}&\textbf{0.04}&~\textbf{0.85}&\textbf{0.00}&~\textbf{0.05}~ \\[1ex]\cline{3-11}
                                 &                             &~0.5~&~$\Theta_5$~&~0.68~&0.08&~0.73&0.04&~0.85&0.00&~0.12~ \\[1ex]\cline{3-11}
																 &                             &~0.7~&~$\Theta_6$~&~0.77~&0.17&~0.73&0.04&~0.85&0.00&~0.21~ \\[1ex]\cline{2-11}
							                   &\multirow{3}{0.5in}{$|_{L}$} &~0.3~&~$\Theta_7$~&~0.75~&0.15&~0.90&0.13&~0.90&0.05&~0.33~ \\[1ex]\cline{3-11}
                                 &                             &~0.5~&~$\Theta_8$~&~0.90~&0.30&~0.90&0.13&~0.90&0.05&~0.48~ \\[1ex]\cline{3-11}
																 &                             &~0.7~&~$\Theta_9$~&~0.90~&0.30&~0.90&0.13&~0.90&0.05&~0.48~ \\[1ex]\cline{2-11}
\hline
\multirow{9}{0.5in}{$\revto_{P}$}&\multirow{3}{0.5in}{$|_{G}$}&~0.3~&~$\Theta_{10}$~&~0.50~&0.10&~0.59&0.18&~0.81&0.04&~0.32~ \\[1ex]\cline{3-11}
                                 &                            &~0.5~&~$\Theta_{11}$~&~0.50~&0.10&~0.59&0.18&~0.81&0.04&~0.32~ \\[1ex]\cline{3-11}
																 &                            &~0.7~&~$\Theta_{12}$~&~0.63~&0.03&~0.59&0.18&~0.81&0.04&~0.25~ \\[1ex]\cline{2-11}
																 &\multirow{3}{0.5in}{$|_{P}$}&~0.3~&~$\Theta_{13}$~&~0.62~&0.02&~0.74&0.03&~0.86&0.01&~0.06~ \\[1ex]\cline{3-11}
                                 &                            &~0.5~&~$\Theta_{14}$~&~0.70~&0.10&~0.74&0.03&~0.86&0.01&~0.14~ \\[1ex]\cline{3-11}
																 &                            &~0.7~&~$\Theta_{15}$~&~0.78~&0.18&~0.74&0.03&~0.86&0.01&~0.22~ \\[1ex]\cline{2-11}
							                   &\multirow{3}{0.5in}{$|_{L}$}&~0.3~&~$\Theta_{16}$~&~0.77~&0.17&~0.90&0.13&~0.90&0.05&~0.35~ \\[1ex]\cline{3-11}
                                 &                            &~0.5~&~$\Theta_{17}$~&~0.90~&0.30&~0.90&0.13&~0.90&0.05&~0.48~ \\[1ex]\cline{3-11}
																 &                            &~0.7~&~$\Theta_{18}$~&~0.90~&0.30&~0.90&0.13&~0.90&0.05&~0.48~ \\[1ex]\cline{2-11}

\hline
\multirow{9}{0.5in}{$\revto_{G}$}&\multirow{3}{0.5in}{$|_{G}$}&~0.3~&~$\Theta_{19}$~&~0.55~&0.05&~0.65&0.12&~0.90&0.05&~0.22~ \\[1ex]\cline{3-11}
                                 &                            &~0.5~&~$\Theta_{20}$~&~0.55~&0.05&~0.65&0.12&~0.90&0.05&~0.22~ \\[1ex]\cline{3-11}
																 &                            &~0.7~&~$\Theta_{21}$~&~0.70~&0.10&~0.65&0.12&~0.90&0.05&~0.27~ \\[1ex]\cline{2-11}
																 &\multirow{3}{0.5in}{$|_{P}$}&~0.3~&~$\Theta_{22}$~&~0.69~&0.09&~0.83&0.06&~0.90&0.05&~0.20~ \\[1ex]\cline{3-11}
                                 &                            &~0.5~&~$\Theta_{23}$~&~0.78~&0.18&~0.83&0.06&~0.90&0.05&~0.29~ \\[1ex]\cline{3-11}
																 &                            &~0.7~&~$\Theta_{24}$~&~0.87~&0.27&~0.83&0.06&~0.90&0.05&~0.38~ \\[1ex]\cline{2-11}
							                   &\multirow{3}{0.5in}{$|_{L}$}&~0.3~&~$\Theta_{25}$~&~0.86~&0.26&~0.90&0.13&~0.90&0.05&~0.44~ \\[1ex]\cline{3-11}
                                 &                            &~0.5~&~$\Theta_{26}$~&~0.90~&0.30&~0.90&0.13&~0.90&0.05&~0.48~ \\[1ex]\cline{3-11}
																 &                            &~0.7~&~$\Theta_{27}$~&~0.90~&0.30&~0.90&0.13&~0.90&0.05&~0.48~ \\[1ex]\cline{2-11}
\hline
\end{tabular}
\end{table}

\section{Discussion}\label{sec-conc}

In this paper, we have been concerned with fuzzy programs belonging to the so-called {\em multi-adjoint logic programming} approach. Our improvements are twofold:
\begin{itemize}
\item On one side, we have extended their syntax for allowing the presence of symbolic weights and connectives on program rules, which very often prevents the full evaluations of goals. As a consequence, we have also relaxed the operational principle for producing what we call {\em symbolic fuzzy computed answers}, where all atoms have been exploited and the maximum  number of expressions involving connectives of the underlaying lattice of truth degrees have been solved too.   

\item On the other hand, we have introduced a tuning process for \ma\
  programs that takes as inputs a set of expected test cases and an
  \sma\ program where some connectives and/or truth degrees are
  unknown.
\end{itemize}
As future work, we consider the implementation of these techniques in
the \fl\ platform, which is available from
\texttt{http://dectau.uclm.es/floper/}.  Currently, the system can be
used to compile \ma\ programs to standard \pr\ code, draw
\emph{derivation trees}, generate declarative traces, and execute \ma\
programs, and it is ready for being extended with powerful
transformation and optimization techniques
\cite{JMP05FSS,JMP06JUCS,JMP09FSS}. Our last update described in
\cite{JMPV15EPTCS,JMPV16RULEML}, allows the system to cope with similarity
relations cohabiting with lattices of truth degrees,   
since this feature is an interesting topic for being
embedded into the new tuning
technique in the near future.


Another interesting direction for further research, consists in
combining our approach with recent fuzzy variants of SAT/SMT
techniques. Research on SAT (Boolean Satisfiability) and SMT
(Satisfiability Modulo Theories)~\cite{barrett2009satisfiability} has
proven very useful in the development of highly efficient solvers for
classical logic. Some recent approaches aim at covering fuzzy logics,
e.g., \cite{Ansotegui2012ISMVL,Godo12}, which deal with propositional
fuzzy formulae containing several propositional symbols linked with
connectives defined in a lattice of truth degrees, as the ones used on
\ma\ programs.\footnote{
  Instead of focusing on satisfiability, (i.e., proving the existence
  of at least one model) as usually done in a SAT/SMT setting, in
  \cite{BMVV13ECEASST,ABLMVV15SUM} we have faced the problem of
  finding the whole set of models
  for a given fuzzy formula by re-using a previous method based on
  fuzzy logic programming where the formula is conceived as a goal
  whose derivation tree, provided by the \fl\ tool, contains on its
  leaves all the models of the original formula, together with other
  interpretations.} We think that our tuning algorithm could
significantly improve its efficiency if the set of \sfca's
instantiated with symbolic substitutions can be expressed as fuzzy
formulae for the fuzzy SAT/SMT solvers mentioned above.


\end{document}